\documentclass{article}

\usepackage[utf8]{inputenc}
\usepackage[left=3cm, right=3cm, top=2cm, bottom=2cm]{geometry}

\usepackage{algorithm}
\usepackage[noend]{algpseudocode}

\usepackage{amsmath}
\usepackage{amssymb}
\usepackage{amsthm}
\usepackage{authblk}

\usepackage{color}

\newtheorem{definition}{Definition}
\newtheorem{theorem}{Theorem}
\newtheorem{corollary}[theorem]{Corollary}
\newtheorem{lemma}[theorem]{Lemma}

\newcommand{\dist}{\text{dist}}

\newcommand{\Althofer}{Alth\"{o}fer}

\newcommand{\Erdos}{Erd\"{o}s}

\title{A Trivial Yet Optimal Solution to Vertex Fault Tolerant Spanners}
\author{Greg Bodwin\thanks{gbodwin6@cc.gatech.edu}\ }
\author{Shyamal Patel\thanks{shyamalpatelb@gmail.com}}
\affil{Georgia Tech}
\date{}

\begin{document}

\maketitle

\begin{abstract}
We give a short and easy upper bound on the worst-case size of fault tolerant spanners, which improves on all prior work and is fully optimal at least in the setting of vertex faults.
\end{abstract}

\section{Introduction}

This paper concerns \emph{spanners}, a fundamental primitive used in graph sketching.

\begin{definition} [Spanners]
A \emph{$k$-spanner} of a graph $G = (V, E, w)$ is a subgraph $H = (V, E' \subseteq E, w)$ for which
$\dist_H(s, t) \le k \cdot \dist_G(s, t)$ for all $s, t \in V$.
\end{definition}
\noindent Spanners have been intensively studied since the mid-80s \cite{PU89jacm, PU89sicomp, ADDJS93, ACIM99, DHZ00, EP04, TZ06, BKMP10, Chechik13soda, AB16stoc, FS16, ABP17}.
A staple in the literature is the \emph{greedy construction algorithm} of \Althofer{} et al \cite{ADDJS93}, which works as follows: initialize $H$ to be an empty graph, and then for each edge $(u, v)$ in the input graph $G$ in order of increasing weight, add $(u, v)$ to $H$ if currently $\dist_H(u, v) > k w_{(u, v)}$.
Besides its simplicity and obvious correctness, the greedy algorithm is popular because there is a simple proof that it is \emph{existentially optimal} \cite{ADDJS93, FS16}, meaning that the number of edges in $H$ never exceeds the worst-case number of edges needed for a $k$-spanner over all possible input graphs on as many nodes as $G$.

In practice, spanners are often applied to systems whose parts are prone to sporadic failures.
A spanner for such a system must be robust to these failures, giving rise to the notion of \emph{fault tolerance}:
\begin{definition} [Fault Tolerant Spanners]
Given a graph $G$, a subgraph $H$ is an \emph{$f$ Vertex Fault Tolerant (VFT), resp.\ Edge Fault Tolerant (EFT),} $k$-spanner of $G$ if for any set of $F$ vertices (resp.\ edges) in $G$ of size $|F| \le f$, $H \setminus F$ is a $k$-spanner of $G \setminus F$.
\end{definition}
\noindent To construct FT spanners, it is natural to consider the obvious adaptation of the greedy algorithm:

\begin{algorithm}
	\caption{The VFT (EFT) Greedy Algorithm}\label{greedy}
	\begin{algorithmic}
	   \Function{ft-greedy}{$G = (V, E, w), k, f$}
	   \State{$H \gets (V, \emptyset, w)$}
	   \For{ $(u,v) \in E$ in order of increasing weight}
	        \If{there exists a set $F$ of $|F| \le f$ vertices (edges) such that $\dist_{H \setminus F}(u, v) > k \cdot w(u,v)$}
	            \State{add $(u, v)$ to $H$}
	        \EndIf
	   \EndFor
       \State \Return{$H$}
       \EndFunction

	\end{algorithmic}
\end{algorithm}

Correctness is again obvious, but unfortunately the analyses used in the non-faulty setting all seem to break for the FT greedy algorithm.
Most prior work on FT spanners has thus abandoned the greedy approach in favor of more involved constructions \cite{LNS98, CZ04, CLPR10, DK11podc, AFIR10} (see also \cite{Parter15, PP13, PP18, BGPV17, CCFK17, DTCR08, BK08, GV12, WY13, DZ16, BGLP16esa}); an analysis of the FT greedy algorithm was only obtained recently via fairly complex arguments \cite{BDPV18}.
In contrast, our main result is a simple analysis of the FT greedy algorithm which improves on all of these previous bounds.

\begin{theorem} [Main Result] \label{thm:main}
Let $b(n,k)$ be the maximum possible number of edges in a graph on $n$ nodes and girth $>k$.
Then any graph $H$ on $n$ nodes returned by the VFT or EFT Greedy Algorithm with parameters $f, k$ satisfies
$$\left|E(H)\right| = O \left (f^2 \cdot b \left (\frac{n}{f}, k+1 \right) \right).$$
\end{theorem}
This upper bound is best possible in the VFT setting, for any construction algorithm, due to a simple lower bound construction in \cite{BDPV18} (meaning that any asymptotic tradeoff between $n, f, k, |E(H)|$ not promised by this theorem does not exist in general).
In the EFT setting, the bound of Theorem \ref{thm:main} was already known to be the best possible tradeoff when $k < 5$ \cite{BDPV18}, but for larger $k$ it is still conceivable to improve the upper bound as far as
$$\left|E(H)\right| \overset{?}{=} O\left(f \cdot b \left (\frac{n}{\sqrt{f}}, k+1\right) + nf\right).$$
It remains a major open question to asymptotically determine $b(n, k)$; the only known upper bound is the folklore \emph{Moore bounds} which state $b(n, k) = O\left(n^{1 + 1/\lfloor k/2 \rfloor}\right)$.
Plugging this into Theorem \ref{thm:main} yields:
\begin{corollary}
For any graph $H$ on $n$ nodes returned by the VFT or EFT Greedy Algorithm with parameters $f, 2k-1$, we have
$$\left|E(H)\right| = O \left( n^{1 + 1/k} f^{1 - 1/k} \right).$$
\end{corollary}
\noindent This corollary improves over the previous best upper bound in \cite{BDPV18} by a factor of $\exp(k)$.
The famous \emph{\Erdos{} girth conjecture} \cite{girth} posits that the Moore Bounds are tight, which would then imply that this corollary is best possible, at least for VFT spanners.

An open question left by this work is to improve the runtime of the FT greedy algorithm: in a naive implementation it is exponential in $f$.
It would be interesting to improve this dependence, or perhaps to find a different fast algorithm achieving the existential size bounds proved in this paper.
We note that \cite{DK11podc} gives a construction with polynomial runtime dependence on $f$, at the price of somewhat suboptimal spanner size.

\section{Proof of Main Result}

We will state the proof in the VFT setting here; the proof in the EFT setting is essentially identical.

\begin{definition} [Blocking Set] Given a graph $G = (V, E)$, a \emph{$k$-blocking set} for $G$ is a set $B \subseteq V \times E$ such that (1) every $(v,e) \in B$ has $v \not \in e$, and (2) for every cycle $C$ in $G$ on $\le k$ edges, there exists $(v,e) \in B$ such that $v,e \in C$.
\end{definition}

\begin{lemma} \label{lem:blocks}
Any graph $H$ returned by the VFT greedy algorithm with parameters $k, f$ has a $(k+1)$-blocking set of size at most $f|E(H)|$.
\end{lemma}
\begin{proof}
For each edge $e = (u,v) \in H$, let $F_e$ be the set of nodes such that $\dist_{H \setminus F_e}(u, v) > kw_e$ when $e$ is added to $H$.
Let
$$B := \left\{ \left(x, e \right) \ \mid \ e \in E(H), x \in F_e \right\};$$
since $|F_e| \le f$ for all $e$, we have $|B| \le f|E(H)|$.
We now show that $B$ is a $(k+1)$-blocking set.
Let $C$ be any cycle on $\le k+1$ edges in the final graph $H$ and let $e = (u,v)$ be the last edge in $C$ considered by the greedy algorithm.
By construction there is a $u \leadsto v$ path (through $C$) of total weight $\le k w_{(u,v)}$ when $e$ is added to $H$, and so some node $x \in C \setminus \{u, v\}$ must be included in $F_e$.
Thus $(x, e) \in B$.
\end{proof}

\begin{lemma} \label{lem:subgraph}
Let $H$ be any graph on $n$ nodes and $m$ edges, let $f = o(n)$ be a parameter, and suppose $H$ has a $(k+1)$-blocking set $B$ of size $|B| \le f |E(H)|$.
Then $H$ has a subgraph on $O(n/f)$ nodes, $\Omega(m/f^2)$ edges, and girth $> k+1$.
\end{lemma}
\begin{proof}

Let $H'$ be an induced subgraph of $H$ on a uniformly random subset of exactly $\lceil n/(2f) \rceil$ vertices, let $B' := B \cap \left(E(H') \times V(H')\right)$, and let $H''$ be obtained from $H'$ by deleting every edge contained in any pair in $B'$.
The graph $H''$ has girth $> k+1$, since by definition of blocking sets we have now deleted at least one node or edge from every cycle in $H$ on $\le k+1$ edges.
Additionally:
\begin{itemize}
\item Each edge in $E(H)$ survives in $E(H')$ iff both of its endpoints survive in $V(H')$, which happens with probability 
$$\frac{\left\lceil n/(2f) \right\rceil}{n} \cdot \frac{\left\lceil n/(2f) \right\rceil - 1}{n - 1} = (1+o(1)) \frac{1}{4f^2}.$$

\item Each pair $(x, (u, v)) \in B$ survives in $B'$ iff all of $x, u, v$ survive in $V(H')$, which happens with probability
$$\frac{\left\lceil n/(2f) \right\rceil}{n} \cdot \frac{\left\lceil n/(2f) \right\rceil - 1}{n - 1} \cdot \frac{\left\lceil n/(2f) \right\rceil - 2}{n - 2} = (1+o(1)) \frac{1}{8f^3} $$
\end{itemize}
We may now compute
$$\mathbb{E}\left[|E(H'')|\right] \ge \mathbb{E}\left[|E(H')| - |B'|\right] = (1+o(1)) \left( \frac{|E(H)|}{4f^2} - \frac{|B|}{8f^3} \right) \ge (1+o(1)) \left(\frac{m}{4f^2} - \frac{m}{8f^2} \right) = \Omega\left(\frac{m}{f^2}\right).$$
There exists a possible setting of $H''$ which matches or exceeds this expectation, which thus has $\Omega(m/f^2)$ edges and satisfies the lemma.
\end{proof}

\begin{proof} [Proof of Theorem \ref{thm:main}]
If $f = \Omega(n)$, then Theorem \ref{thm:main} holds trivially since it states $|E(H)| = O(n^2)$.
Otherwise, let $H$ be an output graph of the FT Greedy algorithm on $n$ edges and $m$ nodes.
By Lemmas \ref{lem:blocks} and \ref{lem:subgraph}, $H$ has a subgraph of girth $>k+1$ on $O(n/f)$ nodes and $\Omega(m/f^2)$ edges, and hence
\begin{align*}
b(O(n/f), k+1) &= \Omega\left(\frac{m}{f^2}\right)\\
m &= O\left(f^2 \cdot b(n/f, k+1)\right). \qedhere
\end{align*}
\end{proof}

We conclude by remarking on a limitation of our approach.
Our definition of blocking sets is very VFT-centric; since a gap remains in the EFT setting, it is tempting to try to adapt this definition to the EFT setting in search of better upper bounds.
In particular, let us say that an \emph{edge $k$-blocking set} is a set of distinct edge pairs such that every $(\le k)$-cycle $C$ has $e_1, e_2 \in C$ for some $(e_1, e_2) \in B$.
It is easy to show that any graph $H$ returned by the EFT greedy algorithm with parameters $f, k$ admits an edge $(k+1)$-blocking set of size $\le f|E(H)|$ (the analog of Lemma \ref{lem:blocks}).
We would then need an \emph{improved} analog of Lemma \ref{lem:subgraph} in order to get improved upper bounds on EFT spanners for $k \ge 5$.
However, no such improvement is possible: for any $k$ we can show a graph $H$ on $\Omega(f^2 \cdot b(n/f, k+1))$ edges that has an edge $(k+1)$-blocking set of size $\le f|E(H)|$, and so our analog of Lemma \ref{lem:blocks} alone is not powerful enough to get improved upper bounds in the EFT setting.
Specifically, this $H$ is the same as the VFT lower bound graph of \cite{BDPV18}: it is the Cartesian product of an arbitrary graph of girth $> k+1$ with a biclique on $\lfloor f/2 \rfloor$ nodes; the blocking set is then all pairs of edges that share an endpoint in the product graph and which correspond to the same edge in the initial high-girth graph.
Hence any improvement to our EFT upper bounds (if possible) will need to exploit stronger properties of $H$ than just the existence of a small edge blocking set.

\bibliographystyle{acm}
\bibliography{References}

\end{document}